\documentclass[11pt]{article}

\usepackage[utf8]{inputenc}
\usepackage{geometry}
\usepackage{fullpage}
\usepackage{amsmath}
\usepackage{url,xspace,hyperref}
\usepackage{amsthm}
\usepackage{amssymb}
\usepackage{amsfonts}
\usepackage{comment}
\usepackage{mathtools}
\usepackage{xcolor}

\usepackage[lined,ruled,vlined]{algorithm2e}

\DeclarePairedDelimiter\ceil{\lceil}{\rceil}

\newtheorem{theorem}{Theorem}

\newtheorem{cl}{Claim}

\newtheorem{lemma}{Lemma}

  \def\rem#1{{\marginpar{\raggedright\scriptsize #1}}}
  
\def\DEBUG{true}
\ifdefined\DEBUG
\newcommand{\lw}[1]{\rem{\textcolor{blue}{$\bullet$ #1}}}
\newcommand{\el}[1]{\rem{\textcolor{gray}{$\bullet$ #1}}}
\newcommand{\vn}[1]{\rem{\textcolor{red}{$\bullet$ #1}}}
\else
 \newcommand{\lw}[1]{}
 \newcommand{\el}[1]{}
 \newcommand{\vn}[1]{}
\fi

\title{On Some Variants of Euclidean K-Supplier}
\author{ Euiwoong Lee\thanks{Computer Science and  Engineering, University of Michigan.} \and Viswanath Nagarajan\thanks{Industrial and Operations Engineering, University of Michigan. Research supported in part by NSF grants CMMI-1940766 and CCF-2006778.} \and Lily Wang$^\dagger$ }

\newcommand{\N}{\mathbb{N}}
\newcommand{\classP}{\mathbf{P}}
\newcommand{\classNP}{\mathbf{NP}}
\newcommand{\R}{\mathbb{R}}
\newcommand{\eps}{\epsilon}

\begin{document}
\maketitle

\begin{abstract}
The    $k$-Supplier problem is an important location problem that has been actively studied in both general and Euclidean metrics. Many of its variants have also been studied,  primarily  on general metrics. 
    We study two variants of $k$-Supplier, namely Priority $k$-Supplier and $k$-Supplier with Outliers, in Euclidean metrics. We obtain    $(1+\sqrt{3})$-approximation algorithms for both variants, which are the first improvements over the previously-known factor-$3$ approximation  (that is known to be best-possible for general metrics). We also study the Matroid Supplier problem on Euclidean metrics, and  show that it cannot be approximated to a factor better than $3$ (assuming $P\ne NP$); so the Euclidean metric offers no improvement in this case.

\end{abstract}

\section{Introduction}
In the $k$-Supplier problem, the input consists of a set of suppliers $I$ and a set of clients $J$ contained in some metric space $(I\cup J, d)$, and $k \in \N$. The goal is to choose a subset $C \subseteq I$ of $k$ suppliers to minimize $\max_{v \in J} d(v, C)$ where $d(v, C) := \min_{u \in C} d(v, u)$. A basic problem in the large and well-studied class of location problems, $k$-Supplier has various applications in operations research including choosing sites for opening plants, placing servers in a network, and clustering data. 
An important special case of $k$-Supplier is $k$-Center where the set of clients $J$ is equal to the set of suppliers  $I$. 

The approximability of $k$-Supplier and $k$-Center on general metric spaces is well understood. A $2$-approximation for $k$-Center and $3$-approximation for $k$-Supplier follow from the work of Gonzalez~\cite{gonzalez1985clustering} and Hochbaum and Shmoys~\cite{hochbaum1985best, hochbaum1986unified}. Simple reductions from Vertex Cover show that these approximation ratios are tight assuming $\classP \neq \classNP$. 

However, the approximability of $k$-Supplier and $k$-Center on  {\em  Euclidean metrics} (which is a practically important special case)  is still open.  Feder and Greene~\cite{feder1988optimal} showed that it is NP-hard to approximate $k$-Supplier and $k$-Center better than $\sqrt{7} \approx 2.65$ and $\sqrt{3} \approx 1.73$ respectively. While it is still open whether one can obtain a $(2 - \eps)$-approximation for $k$-Center for some constant $\eps > 0$, Nagarajan et al.~\cite{nagarajan2013euclidean} obtained  a $(1 + \sqrt{3}) \approx 2.73$ approximation algorithm for Euclidean $k$-Supplier.

Motivated by various practical needs, many variants of $k$-Supplier and $k$-Center also have been proposed and studied in the literature. In the Priority $k$-Supplier problem, the clients are additionally weighted with a priority function $p : V \to \R_+$. Given a set of chosen suppliers $C \subseteq I$, the objective function is now  $\max_{v \in J} p(v) d(v, C)$. This problem naturally models the scenario where each client has a different ``speed''. Plesnik \cite{plesnik1987heuristic} gave a $3$-approximation algorithm for Prioirty $k$-Supplier,  
matching the approximability of the basic version. 

Another variant is $k$-Supplier with Outliers where the input additionally contains a bound $\ell \in \N$ and the goal is to choose $k$ suppliers $C \subseteq I$ and $\ell$ {\em outliers} $O \subseteq J$ to minimize $\max_{v \in J \setminus O} d(v, C)$. This problem was  introduced by Charikar et al.~\cite{charikar2001algorithms}. 
Recently, Chakrabarty et al.~\cite{chakrabarty2020non} obtained a $3$-approximation algorithm for this problem, again 
matching the approximability of the basic $k$-Supplier problem. 
Yet another variant is the Matroid Supplier problem: instead of a cardinality bound on the chosen suppliers, the set $C$ of chosen suppliers is required to be independent in some matroid. Chen et al.~\cite{chen2013matroid} obtained a $3$-approximation algorithm for this problem as well.

\paragraph{Results and Techniques.} 
To the best of our knowledge, the study of the above  $k$-Supplier variants  has been limited to general metrics. 
In this paper, we study these problems in Euclidean metrics.
Our first result is the following: 
\begin{theorem}
There is an $(1+\sqrt{3})\approx 2.73$-approximation algorithm for Euclidean Priority $k$-Supplier.
\label{thm:priority}
\end{theorem}
This is based on a relation to the minimum edge-cover problem, as in \cite{nagarajan2013euclidean}. However, the graph for the edge-cover instance is constructed differently:  we need to select ``representative'' clients (that correspond to nodes in the graph) in decreasing order of their priorities. 

Our second and main technical result is the following: 

\begin{theorem}
There is an $(1+\sqrt{3})\approx 2.73$-approximation algorithm for Euclidean $k$-Supplier with Outliers.
\label{thm:outliers}
\end{theorem}
This requires a linear-program (LP) in conjunction with the relation to edge-cover. Moreover, we do not know how to solve the resulting LP in polynomial time. Instead, we use a ``round or cut'' approach that is built atop the ellipsoid algorithm, and in each step it either finds an approximate solution or a violated LP constraint. 
We note that  round-or-cut  has been used recently to address some other $k$-Supplier problems~\cite{chakrabarty2019generalized}, but the focus there was on general metrics and dealing with   complex constraints on the suppliers. In contrast, our goal is to exploit the Euclidean metric  to  improve the approximation ratio (beyond $3$). Another important step in proving Theorem~\ref{thm:outliers} is an integrality property for the edge-cover polytope with a special type of cardinality constraint; this result might also be of some independent interest.

Finally, we show that not all natural variants of $k$-Supplier  are strictly easier in Euclidean metrics. In particular, we consider the Matroid Supplier problem where there is a matroid constraint on $I$ and the goal is to find an independent set $C$ that minimizes $\max_{v \in J} d(v, C)$. While this problem admits a $3$-approximation algorithm in general metrics~\cite{chakrabarty2019generalized}, we prove the following theorem that Euclidean spaces do not strictly improve the approximation ratio.  

\begin{theorem}
For any constant $\eps > 0$, it is NP-hard to approximate Euclidean Matroid Supplier within a factor of $(3 - \eps)$. 
\label{thm:hardness}
\end{theorem}

\paragraph{Other Related Work.}
Apart from $k$-Supplier/$k$-Center, such variants have also been studied for $k$-Median (where the objective is to minimize the sum of connection costs). In particular, there are constant-factor approximation algorithms for $k$-Median with outliers~\cite{Chen08,KrishnaswamyLS18} and Matroid Median~\cite{Krishnaswamy0NS15,Swamy16}.  Moreover, there is an extensive literature on obtaining better approximation ratios (and runtime) for $k$-Median on Euclidean metrics, see e.g., \cite{KR07,HM04}.

\def\sse{\subseteq}
\section{$k$-Supplier with Priorities}

Given a set of suppliers $I$ and clients $J$, where  clients have a priority function $p: J \to \mathbb{R}_+$,  the goal is to choose $k$ suppliers to minimize the maximum ``priority weighted distance'' over all clients. That is, we want to find \[\min_{\substack{C\subseteq I\\|C| \leq k}} \quad \max_{v\in J} \,\, p(v) \cdot d(v, C). \]
For a given set of suppliers $C$, the {\em priority distance} of any client $v\in J$ is $p(v)\cdot d(v,C)$.

\paragraph{Assuming optimal value of $1$.} As is common for min-max optimization problems (see e.g., \cite{hochbaum1985best}), we assume that the algorithm knows the optimal value  $B$. Then, the algorithm either finds a solution of objective at most $\alpha\cdot B$ (where $\alpha$ is the approximation ratio), or proves that the optimal value is more than $B$. As there are only a polynomial number of choices for $B$, we can try each one.  Finally, by  scaling all distances by $B$, we can assume that the optimal value is  $1$. 

\medskip

Our  algorithm is similar to that in \cite{nagarajan2013euclidean} for the basic $k$-Supplier. This involves constructing a graph with some clients $S\sse J$ as nodes and suppliers as edges, and finding the minimum edge-cover in this graph. 
The key difference is that we need to include clients into the node-set $S$ in decreasing order of priorities. See Algorithm~\ref{alg:priority} for details.

\begin{algorithm}
\caption{Algorithm for Priority $k$-Supplier\label{alg:priority}}
\LinesNumbered
initially nodes $S = \emptyset$ and edges $E =\emptyset$\;
\While{$J \neq \emptyset$}{
$\bar{v} = \arg\max_{v \in J} p(v)$\;
$E_{\bar v} \gets \{v \in J: p(v)\cdot d(v,\bar v) \leq \sqrt{3}\}$\;
$J \gets J \setminus E_{\bar v}$ and $S \gets S \cup \{\bar{v}\}$\;
}
\ForAll{supplier $u \in I$}{
\uIf{$\exists$ distinct $\bar v_1, \bar v_2 \in S \text{ s.t. } p(\bar v_1) d(u, \bar v_1) \leq 1$ \text{ and } $p(\bar v_2) d(u, \bar v_2) \leq 1$}{
add edge $(\bar v_1, \bar v_2)$ to $E$ and  label it $u$\;}{}
\uElseIf{$\exists \bar v \in S \text{ s.t. } p(\bar v) d(u, \bar v) \leq 1$}{
add self-loop to $(\bar v, \bar v)$ to $E$ and   label it $u$\;}
}
Find the  minimum edge cover $M$ in graph $(S,E)$\;
\uIf{$|M|\le k$}{output the suppliers labeled on edges of $M$\;}
\uElse{the optimal value is more than $1$\;}
\end{algorithm}

For the analysis, we will show that if the optimal value is at most $1$, the algorithm returns solution $M$ with objective at most $1+\sqrt{3}$. Henceforth, we assume that the optimal value is at most $1$.

\begin{lemma}
Each client $v \in J$ is within priority-distance $(1+\sqrt{3})$ from some supplier in $M$.
\end{lemma}
\begin{proof}
Consider any $v \in J$: it must lie in   $E_{\bar v}$ for some ``selected'' client ${\bar v}\in S$. Note that  $\bar v$ must be covered by some edge in $M$, say labelled by supplier $u$. Then, 
\[p(v)d(v,u) \leq p(v)d(v, \bar v) + p(v)d(\bar v, u) \leq p(v)d(v, \bar v) + p(\bar v)d(\bar v, u) \leq \sqrt{3} + 1.\] 
The second inequality uses the fact that at the point when $\bar v$ was added to $S$, client  $v$ was also in $J$: so $p(\bar{v})\ge p(v)$. The  third inequality is by definition of $E_{\bar v}$ and edges $E$.
\end{proof}

\begin{lemma}\label{lem:priority-2}
No supplier can serve more than two clients of $S$ within priority-distance $1$.
\end{lemma}
\begin{proof}
Suppose for a contradiction that for supplier $u \in I$ and clients $v_1, v_2, v_3 \in S$ are within priority-distance $1$ from $u$. Then, we have $p(v_i)d(v_i, u) \leq 1$ for $i = 1,2,3$. 
There is at least one pair of $v_i, v_j$ such that the angle $\theta = \angle v_i u v_j\le 2\pi/3$. See Figure~\ref{fig:priority}. Suppose without loss of generality that $p(v_i) \geq p(v_j)$, so that $v_i$ is added to $S$ before $v_j$. By the cosine law, 
\begin{align*}
    d(v_i, v_j) &= \sqrt{d(v_i,u)^2  +d(v_j,u)^2  - 2\cdot d(v_i,u)\cdot d(v_j,u)\cdot \cos \theta} \\
&\le \sqrt{\frac{1}{p(v_i)^2} + \frac{1}{p(v_j)^2} + \frac{1}{p(v_i)p(v_j)}} \leq \sqrt{\frac{3}{p(v_j)^2}} = \frac{\sqrt{3}}{p(v_j)} 
\end{align*}
It follows that \[p(v_j)d(v_j, v_i) \leq p(v_j) \frac{\sqrt{3} }{p(v_j)} = \sqrt{3}.\]
Therefore, $v_j$ should have been in $E_{v_i}$ and can not be in $S$, a contradiction.
\end{proof}

\begin{figure}
    \centering
    \includegraphics[scale=0.3]{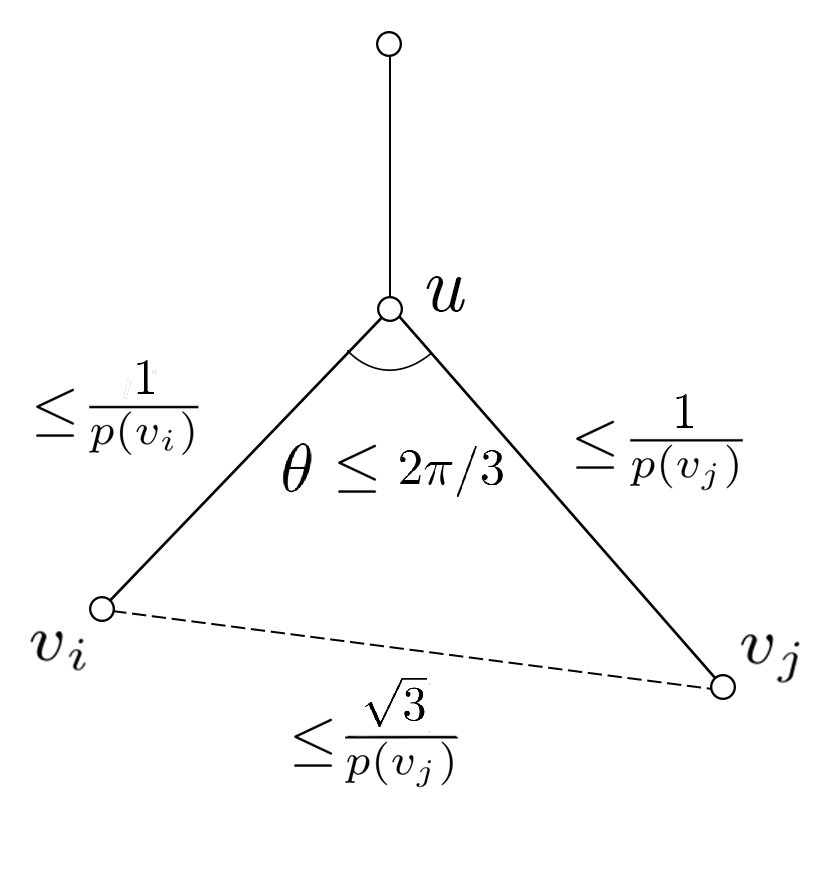}
    \caption{Illustration of Lemma~\ref{lem:priority-2}.}
    \label{fig:priority}
\end{figure}

\begin{lemma}
The minimum edge cover $M$ satisfies $|M| \leq k$.
\end{lemma}
\begin{proof}
Let $M^*\sse I$ be the optimal set of suppliers. Note that $M^*$   covers each client within priority distance $1$. Moreover, by Lemma~\ref{lem:priority-2}, each supplier can cover at most two clients of $S$ within priority distance $1$. In other words, taking the edges corresponding to the suppliers $M^*$ in graph $(S,E)$, we get an edge cover. Therefore, the minimum edge cover $M$ has size at most $|M^*|=k$.
\end{proof}
Combining the lemmas above, we obtain Theorem~\ref{thm:priority}.

\section{$k$-Suppliers with Outliers}
Here, we are given a set of suppliers $I$ and clients $J$ along with bounds  $k$ on the number of chosen  suppliers and $\ell$ on the number of outlier clients. As mentioned earlier, we assume that the optimal value is $1$, and aim to find a solution with objective at most $1+\sqrt{3}$. This would prove  Theorem~\ref{thm:outliers}.

We start with a natural  LP relaxation  where decision variables  $y_i$ correspond to selecting  suppliers and $z_j$ correspond to choosing outlier clients.
\begin{align}
& \sum_{i \in I} y_i \le k & \label{lp:start}\\
    & z_j + \sum_{i \sim j} y_i \geq 1 & \forall j \in J\\
    & \sum_{j \in J}z_j \leq \ell& \label{lp:outlier-constraint} \\ 
    &0 \leq z, y \leq 1 \label{lp:end}
\end{align}
Above,   $i \sim j$ denotes client $j$ being within unit distance from supplier $i$, i.e., supplier $i$ can serve client $j$. While these constraints suffice to obtain a $3$-approximation algorithm (even on general metrics), we need to add stronger constraints for the improved $1+\sqrt{3}$ approximation ratio. 

Define a subset of clients $S\sse J$ to be {\em well-separated} if all pairwise distances in $S$ are greater than $\sqrt{3}$, i.e., $d(j, j') > \sqrt{3}$ for every $j, j' \in S$. Also, for any set of clients $S$, we will denote the set of suppliers  which can  serve at least one client in $S$ by $f(S) \subseteq I$. 
The stronger constraints we want to add are the following: 
\begin{equation}\label{eq:well-sep-constraint}
    z(S) + y(f(S)) \geq \ceil{|S|/2} \qquad \forall S\sse J \text{ well-separated}.
\end{equation}
Above, we use the shorthand $z(S):=\sum_{j\in S} z_j$ and $y(f(S)):=\sum_{i\in f(S)} y_i$.

We now show that  these constraints are valid for any (integral) solution to $k$-Supplier with Outliers. Consider any well-separated set $S$. Note that no supplier can serve more than two clients in $S$: this follows from Lemma~\ref{lem:priority-2} with all priorities being $1$ (or Lemma~1 in \cite{nagarajan2013euclidean}). Hence, a  total of at least  $ \ceil{|S|/2} $  suppliers from $f(S)$ or outliers in $S$ are needed to ``cover''  the clients in $S$.

Our final LP relaxation, referred to as the 
``Master LP'' consists of constraints \eqref{lp:start}-\eqref{lp:end} and \eqref{eq:well-sep-constraint}. There are an exponential number of  well-separated constraints, and we are not aware of a separation oracle for these. So,  this LP  is difficult to solve directly. Instead, we will use a round-or-cut approach that either (i)  finds a solution of objective at most $1+\sqrt{3}$, or (ii) proves that the Master LP is infeasible. Note that  case (ii) also implies that the optimal value of the $k$-Supplier with Outliers problem is more than $1$. So this would suffice to prove Theorem~\ref{thm:outliers}.

We are now ready to describe the algorithm, which relies on the ellipsoid algorithm with separation-oracles. We will maintain a candidate solution $(y,z)$ for the Master-LP, and an ellipsoid ${\cal F}$ that is guaranteed to contain Master-LP. 

In each iteration below, we either (i) find an approximate solution to $k$-Supplier with Outliers, or (ii) identify a violated constraint for the Master-LP, which is used to update our solution $(y,z)$ and the ellipsoid ${\cal F}$.  
Formally, we repeat the following steps. 
\begin{enumerate}
\item If $(y,z)$ violates any of the (polynomially many)  constraints \eqref{lp:start}-\eqref{lp:end}, then update   solution $(y,z)$ and   ellipsoid ${\cal F}$ based on the violated constraint. Continue to the next iteration.
    \item Let nodes $A = \emptyset$, $B = J$.
    \item Order clients by outlier values from the LP solution: $z_1 \le z_2 \le \ldots \le z_n$, where $n = |J|$.
    \item \label{rep} 
    While $B \neq \emptyset$ do: 
    \begin{itemize}
        \item Let $j \in B$ with the client with lowest $z_j$.
        \item Take all clients in $B$ within distance $\sqrt{3}$ of $j$ (including itself) and assign them to $R_j$.
        \item Remove $R_j$ from the set $B$.
        \item Let $a(j) = |R_j|$ denote the number of clients assigned to $j$.
        \item Add node $j$ to $A$.
    \end{itemize}
    \item Construct a graph $G$ with nodes $A$  and the following edges. For each supplier $i$: 
    \begin{itemize}
        \item If there are two distinct clients $j_1,j_2\in A$ within distance $1$ from $i$, add edge $(j_1,j_2)$ labelled by $i$.
        \item Otherwise, if there is just one client $j\in A$ within distance $1$ from $i$, add self-loop $(j,j)$ labelled by $i$.
    \end{itemize}
    Let  $E$ be the set of all  edges added above. Furthermore, add a distinct set $L$ of self-loops at each vertex $j \in A$: the loop at $j$ represents making $j$ an outlier. All edges of $E$ have  weight $0$. Each loop $(j,j)$ in $L$ has  weight $a(j)$.
    
    \item Check whether $(y, z)$ satisfies the following constraints:
    \begin{equation}\label{eq:edge-cover-constraints}
        z(S) + y(f(S)) \geq \lceil |S| / 2 \rceil \qquad \forall S \subseteq A.
    \end{equation}
    These constraints exactly specify the edge cover polytope of graph $G=(A, E\ \dot\cup \ L)$ and can be efficiently checked~\cite{schrijver2003combinatorial}. 
   
    \item \label{case1} If $(y, z)$ violates \eqref{eq:edge-cover-constraints} for some $S \subseteq A$, then:
    \begin{itemize}
        \item Update solution $(y,z)$ and the ellipsoid ${\cal F}$ based on the constraint for $S$. (Note that the constraint for $S$ appears in \eqref{eq:well-sep-constraint} of the Master-LP as $S\sse A$ is well-separated.)
        \item Continue to the next iteration.
    \end{itemize}

    \item \label{case2} If $(y,z)$ satisfies \eqref{eq:edge-cover-constraints}, apply Theorem~\ref{thm:EC-card} below to obtain a solution $M$ to min-weight edge-cover on graph $G$ with a cardinality constraint on $E$. Output the suppliers in $M\cap E$ as the approximate solution, and stop.
    
   \end{enumerate}

Assuming that the algorithm never stops in step~8, the standard analysis for the ellipsoid algorithm (see e.g., \cite{gls1988}) implies that we can terminate after a polynomial number of iterations and  conclude that the Master-LP is infeasible. Therefore, the overall algorithm is guaranteed to run in polynomial time. Moreover, we either return some solution $M$ (in step 8) or prove that the Master-LP is infeasible. In the analysis below, we will show that the solution $M$ obtained in step~8 is a $1+\sqrt{3}$ approximation for  $k$-Supplier with Outliers.

\paragraph{Edge cover with a cardinality constraint.} Consider a graph $G$ on nodes $A$ and edges $E' = E \ \dot\cup\ L$, where $L$ only contains self-loops. (Edges in $E$ can be 2-edges or self-loops.) Note that we use the same notation as for the graph constructed in step~5 of the above algorithm. Each edge $e\in E'$ has a weight $w_e$. We are interested in solving the minimum weight edge-cover problem on $G$ subject to a  cardinality constraint of $k$ on $E$. That is, we want a min-weight edge cover $M\sse E'$ where $|M\cap E|\le k$. Note that the cardinality constraint does not include   all edges $E'$, but only those in set $E$. We will show that this problem can be solved in polynomial time using the natural LP relaxation.  

Consider the following linear program $LP_{ECC}$ for the above  edge-cover problem with a  cardinality constraint. Recall that the  edges are  $E' = E \ \dot\cup\ L$. We use decision variables  $y \in \R^E$ for the edges in $E$ and $z \in \R^L$ for the remaining edges $L$. 
\begin{align*}
     & \sum_{e\in E} y_e \leq k &(LP_{ECC})\\ 
    & z(S) + y(f(S)) \geq \ceil{|S|/2} &\forall S \subseteq A\\
    &z, y \geq 0&    
\end{align*}

\begin{theorem}\label{thm:EC-card}
$LP_{ECC}$ is integral. 
Moreover, there is a polynomial time algorithm for the min-weight edge-cover problem with a cardinality constraint.

\end{theorem}

We defer the proof of this theorem to Section~\ref{subsec:LPecc}.

\paragraph{Completing the proof of  Theorem~\ref{thm:outliers}.}
 We now use Theorem~\ref{thm:EC-card} to show that the solution $M$ found in step 8 of our algorithm is a feasible solution to $k$-Supplier with Outliers of objective at most $1+\sqrt{3}$. 
 
 When the algorithm reaches step~8, observe that all  constraints in \eqref{eq:edge-cover-constraints}
 are satisfied by the current solution $(y,z)$. Moreover, by step~1, all the basic  constraints~\eqref{lp:start}-\eqref{lp:end} are also 
 satisfied. It follows that this solution $(y,z)$
 is also feasible for $LP_{ECC}$. By definition of the edge-weights in the edge-cover instance, the weight objective  of this solution is:
 $$\sum_{e\in E}w_e\cdot y_e + \sum_{(j,j)\in L}w_{(j,j)}\cdot z_j = \sum_{j\in A} a(j)\cdot z_j = \sum_{j\in A} |R_j| \cdot z_j\le \sum_{j\in A} \sum_{j'\in R_j} z_{j'} \le \sum_{j'\in J} z_j \le \ell.$$ 
The first inequality uses the fact that we select clients into $A$ in  increasing order of $z$-values: so $z_j\le z_{j'}$ for all $j'\in R_j$. The second inequality uses that $\{R_j : j\in A\}$ are disjoint. The last inequality uses constraint~\eqref{lp:outlier-constraint}. Therefore, the integral solution $M$ to  $LP_{ECC}$ (found by Theorem~\ref{thm:EC-card}) has weight  $\sum_{e\in M} w_e = \sum_{(j,j)\in M\cap L} a(j) \le \ell$. Let $A'\sse A$ denote the clients/nodes in graph $G$ that are covered by the edges $M\cap E$. Note that every client in $A'$ is within distance $1$ from some supplier of $M\cap E$. Hence, every client in $\cup_{j\in A'} R_j$ is within distance $\sqrt{3}+1$ from $M\cap E$.
Moreover, $M\cap L$ must contain the loops at each of the clients $A\setminus A'$. It then follows that $\sum_{j\in A\setminus A'} |R_j|  \le  \sum_{(j,j)\in M\cap L} a(j) \le \ell$. We set $O=\cup_{j\in A\setminus A'} R_j$ to be the outlier clients. From the above discussion, it is clear that each non-outlier client is within distance $\sqrt{3}+1$ from $M\cap E$ and the number of outliers   $|O|\le \ell$. Finally,  $|M\cap E|\le k$ because of the cardinality constraint. It now follows that 
$M\cap E$ is a feasible solution  to $k$-Supplier with Outliers of objective at most $1+\sqrt{3}$.

\subsection{Proof of Theorem~\ref{thm:EC-card}}
\label{subsec:LPecc}

We note that if the set $L=\emptyset$ (i.e., the cardinality constraint involves {\em all} edges) then $LP_{ECC}$ is known to be integral: see the discussion in page 464 of \cite{schrijver2003combinatorial}. However, this does not directly imply Theorem~\ref{thm:EC-card}. Moreover, the following example shows that Theorem~\ref{thm:EC-card} is not true for a cardinality constraint on an {\em arbitrary} edge subset. Hence, our proof below relies crucially on the fact that $L$ only contains self-loops.

\paragraph{Example:} suppose graph $G$ is a 
4-cycle with edges $a, b, c, d$ in that order. The cardinality constraint is imposed on  $E=\{a,c\}$, with a bound of $k=1$. Note that any integral solution to $LP_{ECC}$ must be of the form $(0,\alpha, 0,\beta)$,  $(1,\alpha, 0,\beta)$ or  $(0,\alpha, 1,\beta)$, where $\alpha,\beta \in \mathbb{Z}_{\ge 1}$.  
It can be checked directly that  the solution  $(\frac{1}{2}, \frac{1}{2}, \frac{1}{2}, \frac{1}{2})$  cannot be written as a convex combination of integer solutions, which shows that $LP_{ECC}$ is not integral for this instance.

\medskip

Recall that the set of edges is $E' = E \ \dot\cup\ L$, 
$y \in \R^E$, and $z \in \R^{L}$, where the set $L$ only contains self-loops. ($E$ may contain self-loops too.)  
For any {\em multi-subset} $S\subseteq E'$ of edges, we use $\mathbf{1}(S)\in \mathbb{Z}^{E'}$ to denote the  vector of multiplicities.
Before proving Theorem~\ref{thm:EC-card}, we show the following key lemma.

\begin{lemma}\label{lem:decomp}
Consider any feasible solution $(y,z)$ for $LP_{ECC}$. There is a collection $\{J_i\}_{i=1}^r$ of integral solutions (i.e, edge covers that satisfy the cardinality constraint) and convex multipliers $\{\lambda_i\}_{i=1}^r$ such that $(y,z) \ge \sum_{i=1}^r \lambda_i \cdot \mathbf{1}(J_i)$. 
\end{lemma}
\begin{proof} 
Fix any fractional solution $(y, z)$ to $LP_{ECC}$. Clearly, this is also feasible to the basic edge-cover LP (without the cardinality constraint). By integrality of the edge-cover LP (Theorem 27.3 of \cite{schrijver2003combinatorial}), it follows that  $(y,z)$ dominates   a convex combination of  integral edge-covers. Let $(y, z) \ge \sum_{a=1}^r \lambda_a \cdot \mathbf{1}(J_a)$ denote such a convex combination where the $J_a$ are integral edge-covers and the $\lambda_a$ are convex multipliers. Over all such possible convex combinations, choose the one which produces the least ``variance'' as measured by 
\[\sum_{a=1}^r \lambda_a \cdot \max \left(0, |E\cap J_a|  - k \right).\]  
We can assume (without loss of generality) that each integral edge-cover $J_a$ is minimal. Indeed, if $J_a$ is not minimal, we can replace it by a minimal  edge-cover $\bar{J_a}\subsetneq J_a$: the variance of the resulting convex combination can only decrease.

If the variance is $0$ then we must have $| E\cap J_a|  \leq k$ for every $a$, which implies that each $J_a$ is an  integral solution to $LP_{ECC}$. In this case, the lemma is trivially true.

We now suppose (for a contradiction)   that the variance is positive. As the variance is positive,   we have some $i\in [r]$ with $| E\cap J_i| \geq k + 1$ by integrality. As $(y,z)$ satisfies the cardinality constraint, we have $\sum_{a=1}^r \lambda_a | E\cap J_a|  \le \sum_{e\in E} y_e\le k$. Therefore, there is some   $\ell\in [r]$ with $|E\cap J_\ell|\leq k - 1$ (again by integrality).

Let $C = E\cap J_i$ and $C_o= L\cap J_i$. 
Note that both $C$ and $C_o$ are sets (not multisets) because of minimality of $J_i$. Likewise, let $D= E\cap J_\ell$ and $D_o= L\cap J_\ell$.

We now  convert edge-cover  $J_i = C\dot\cup C_o$ into  a perfect matching  (with loops) as follows.
\begin{enumerate}
    \item Let $ \bar{C}\sse C$ be any maximal matching using only 2-edges. 
    \item Then, for any other edge $e=(u,v)\in (C\cup C_o)\setminus \bar{C}$, if one of its  nodes (say $u$) is incident to the matching $M$ then we modify $e$ into the self-loop $(v,v)$; otherwise edge $e$ remains unchanged. Note that  such an edge $e$ cannot have both nodes $u,v$  incident to matching $M$, by  minimality of edge-cover $C\cup C_o$. 
\end{enumerate}
Let $\bar{C_o}$ denote all edges created  in step 2 above. Note that $\bar{C}\cup \bar{C_o}$ is a perfect matching: each node has exactly one edge (either 2-edge or self-loop) incident to it. Note also that there is a 1-to-1 correspondence between the edge-covers $J_i=C\cup C_o$ and  $\bar{C}\cup \bar{C_o}$. See Figure~\ref{fig:ec-match} for an example.
\begin{center}
\begin{figure}
    \centering
    \includegraphics[scale=0.2]{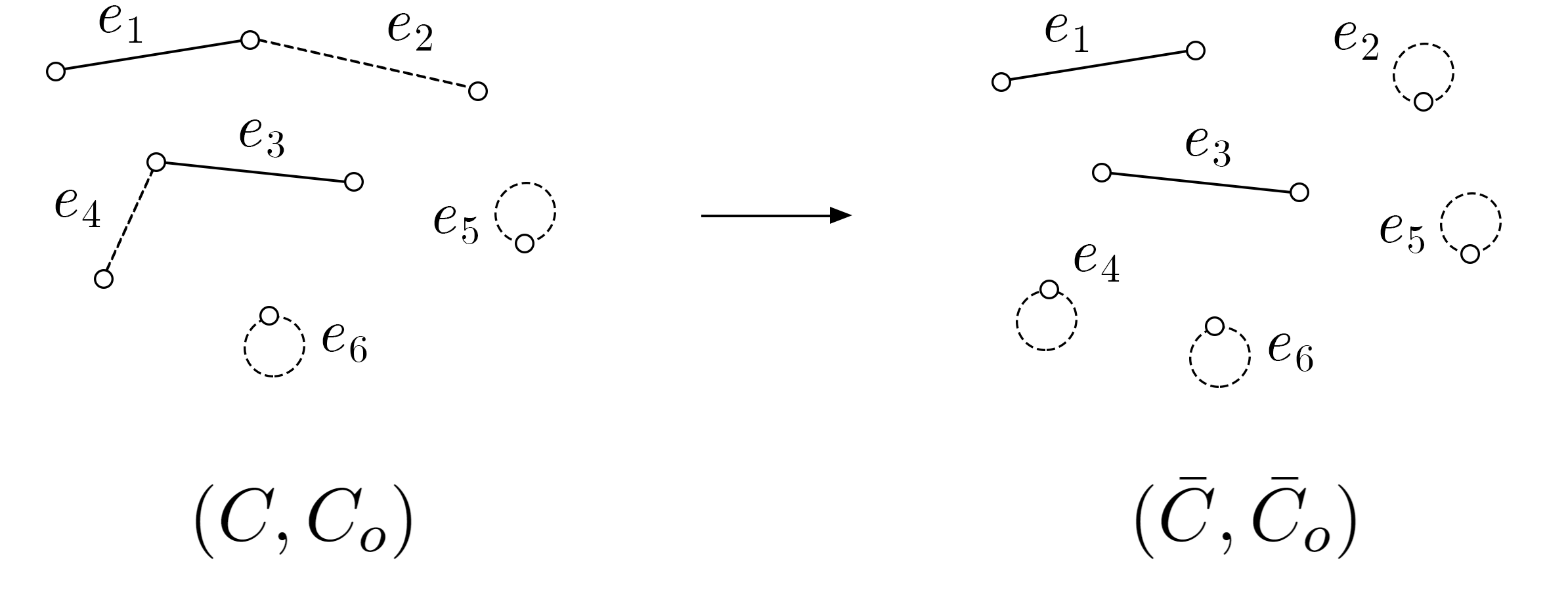}
    \caption{Converting edge-cover to a perfect matching.     \label{fig:ec-match}  }
\end{figure}
\end{center}

We apply the same procedure to modify  edge-cover 
 $J_\ell = D\dot\cup D_o$ into  $(\bar{D},\bar{D_o})$. We now have two graphs, each of which is a perfect matching (with self-loops). 
 Let ${\cal G}$ denote the disjoint  union $\bar{C}\dot\cup \bar{C_o} \dot\cup\bar{D}\dot\cup \bar{D_o}$ of all these edges. Note that each connected component in ${\cal G}$ is either an even cycle (with 2-edges) or a path with self-loops at both ends. 
 
 Assign a value of $1$ (resp. $-1$) to all edges in $\bar C \dot\cup \bar{C_o} $ (resp. $\bar D\dot\cup \bar{D_o} $) that correspond to $E$-edges. All the other edges (corresponding to $L$-edges) are assigned value $0$. Note that every 2-edge has value $+1$ or $-1$. Over the entire graph, the total value is positive as 
 $$ |(\bar C\dot\cup \bar{C_o})\cap E| =|C|  > |D|= |(\bar D\dot\cup \bar{D_o})\cap E|.$$
 So, there is some component $H$ in ${\cal G}$ with positive total value. Note that component $H$ cannot be a cycle: any cycle is even and hence has value $0$. So $H$  is a path with self-loops at both ends. (The path may also be empty, in which case we have a node with two self-loops.) Moreover,  the 2-edges on the path have alternating positive or negative value. The self-loops at the end of path $H$ have either $0$ value or the opposite sign as the 2-edge they're incident to. Hence,  component $H$ has total value $-1$, $0$, or $1$. Since it has positive value, it must have value exactly  $1$. We now define two new edge-covers: $X$ (resp. $Y$) consists of the edges from $J_i$ (resp. $J_\ell$) in all components except $H$, and edges from $J_\ell$  (resp. $J_i$) in component $H$. Note that $X$ and $Y$ are indeed edge covers. Moreover, $|X\cap E| = |J_i\cap E|-1 = |C|-1$ and $|Y\cap E| = |J_\ell\cap E|+1 = |D|+1$.

\begin{figure}[h]
    \centering
    \includegraphics[scale=0.15]{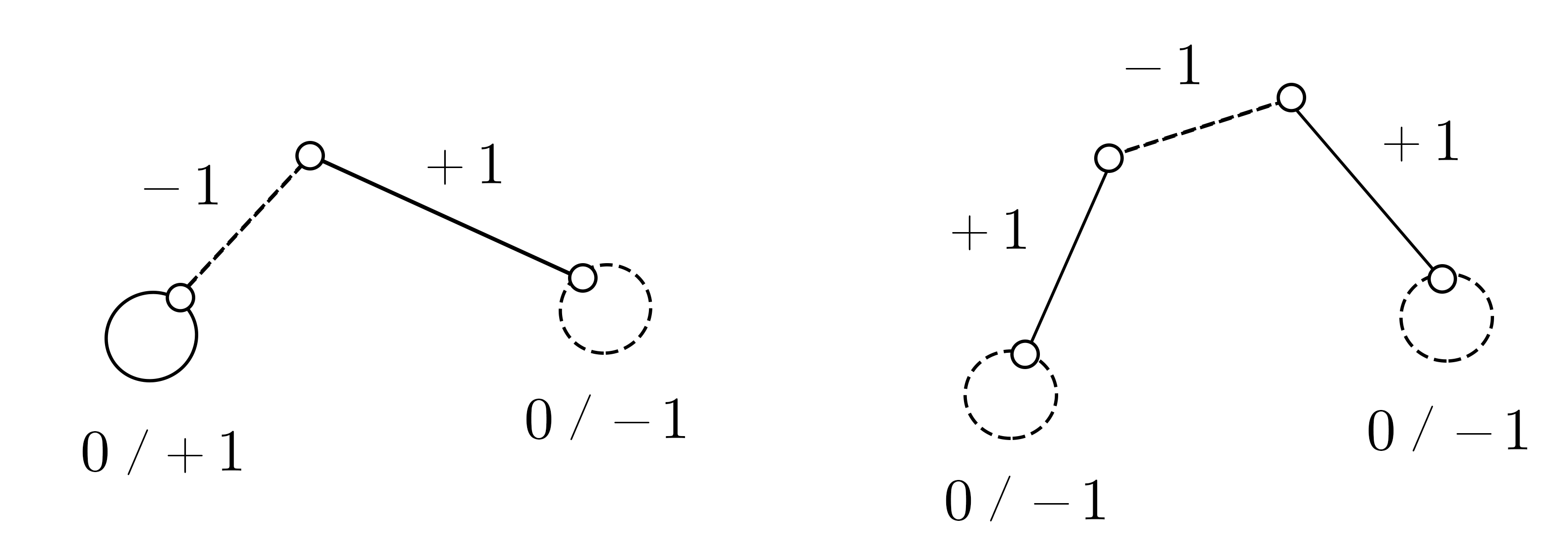}
    \caption{Cases for a component $H$ with positive value}
    \label{fig:my_label}
\end{figure}

 We now construct a new convex combination that has smaller variance, which leads to a contradiction. Recall the edge-covers $\{J_a\}_{a=1}^r$ in the original convex combination. Let   $J_{r+1}=X$ and $J_{r+2}=Y$ be the two new edge-covers. Let $\epsilon = \min\{\lambda_i , \lambda_\ell\}>0$. The convex multipliers are now:
 $$\lambda'_a = \left\{ \begin{array}{ll}
 \lambda_a-\epsilon       &  \mbox{ if } a=i,\ell \\
     \epsilon & \mbox{ if } a=r+1, r+2\\
     \lambda_a & \mbox{otherwise}
 \end{array}\right. ,\quad \forall a\in [r+2].$$
Clearly, $\sum_{a=1}^{r+2}\lambda'_a\cdot \mathbf{1}(J_a) = \sum_{a=1}^{r}\lambda_a\cdot \mathbf{1}(J_a) \le (y,z)$. We now bound the increase in  variance: 
\begin{align*}
    &\sum_{a = 1}^r (\lambda_a' - \lambda_a) \max \left(0, |E\cap J_a| - k \right)  + \epsilon\cdot \max(0, |C|-1-k) + \epsilon\cdot \max(0, |D|+1-k) \\
    =& -\epsilon \cdot \max(0, |C|-k) - \epsilon \cdot \max(0, |D|-k) + \epsilon\cdot \max(0, |C|-1-k) + \epsilon\cdot \max(0, |D|+1-k) \\
    \le &-\epsilon,\end{align*}
where the last inequality uses the fact that $|C|-1-k\ge 0 \ge |D|+1-k$. As $\epsilon>0$, this contradicts the choice of the original convex combination (of minimum variance). This completes the proof. 
\end{proof}

\begin{proof} [Continued Proof of Theorem \ref{thm:EC-card}] We first prove the integrality of $LP_{ECC}$. Given any fractional solution $(y,z)$ to $LP_{ECC}$, Lemma~\ref{lem:decomp} implies $(y, z) \geq \sum_{i=1}^r \lambda_i \cdot \mathbf{1}(J_i)$ for some convex combination of integral solutions. We now show that we can ensure equality, i.e., $(y,z)$ is equal to a convex combination of integral solutions. Clearly, this would prove that  $LP_{ECC}$ is integral.

We can write $(y, z) = \sum_{i=1}^r \lambda_i \cdot \mathbf{1}(J_i) + (y', z')$
for some $y' \in \R^{E}, z' \in \R^{L}$ with $y', z' \geq 0$. Note that  $k \geq y(E) = \sum_{i=1}^r \lambda_i |J_i \cap E| + y'(E)$. So, if $y'(E) > 0$, there exists $i \in [r]$ such that $|J_i \cap E| < k$. Choose an edge $e \in E$ such that $y'_e > 0$. We now perform one of the following modifications that maintains $(y, z) = \sum_{i=1}^r \lambda_i \cdot \mathbf{1}(J_i) + (y', z')$ and $y', z' \geq 0$ while strictly decreasing $y'(E)$.

\begin{enumerate}
    \item If $\lambda_i > y'_e$, let 
    $\lambda_i \leftarrow \lambda_i - y'_e$ and create a new index $r+1$ such that $\lambda_{r+1} = y'_e$ and $J_{r+1} = J_i \ \dot\cup\ \{ e \}$. Let $r \leftarrow r + 1$ and $y'_e = 0$. 
    \item If $\lambda_i \leq y'_e$, let $J_i \leftarrow J_i \ \dot\cup\ \{ e \}$ and $y'_e \leftarrow y'_e - \lambda_i$. 
\end{enumerate}

The above step 1 strictly decreases the support of $y'$, so cannot be done more than $|E|$ times. Between two consecutive applications of step 1's, each application of step 2 strictly increases the size of one $J_i$: so it can  be done at most $rk$ times. (And $r$ increases by at most one for each application of step 1.) Therefore, the above procedure can be repeatedly applied and finished in finite time so that $y'(E) = 0$ at the end. The same procedure can be applied for $z'$ as well, which is even easier because we do not have the cardinality constraint for $L$. At the end, we have $(y, z) = \sum_{i=1}^r \lambda_i \cdot \mathbf{1}(J_i)$ where each $J_i$ is an integral edge cover that satisfies the cardinality constraint. We note that these edge-covers $J_i$ may be multisets (and not minimal edge covers).

To obtain a polynomial time algorithm for min-weight edge-cover with a cardinality constraint, we first solve $LP_{ECC}$ optimally using the ellipsoid algorithm. This can be done because there is an efficient separation oracle for the edge-cover LP. The resulting solution $(y,z)$ may not be an extreme point of $LP_{ECC}$ (and hence not integral). However, we can apply a standard polynomial-time method for converting an arbitrary LP solution into an extreme point solution (assuming  a separation oracle for the constraints); see e.g., Lemma 3.3 in \cite{Jain}. Hence, we can find an optimal  extreme point solution $(y^*,z^*)$ to $LP_{ECC}$ in polynomial time. By integrality of $LP_{ECC}$,  $(y^*,z^*)$ is an integral  optimal solution. 
\end{proof}

\section{Hardness for Matroid Supplier}
We now consider the  Euclidean Matroid Supplier problem. Its input consists of $I \cup J \subseteq \R^{s}$ and a matroid $\mathcal{I}$ on ground set $I$, and the goal is to find an independent set $C \in \mathcal{I}$ that minimizes $\min_{j \in J} d(j, C)$, where $d$ denotes the Euclidean distance. We prove that this problem is $(3-\eps)$-hard to approximate for any constant $\eps > 0$, proving Theorem~\ref{thm:hardness}.

We reduce from the NP-hard 1-in-3-SAT problem \cite{schaefer1978complexity}. This involves $n$ binary variables and $m$ clauses, each consisting of three literals (of any variable or its negation). The goal is   to decide whether there is an  assignment where {\em  exactly} one literal is true in each clause. 

Suppose that we have a $(3 - \epsilon)$-approximation algorithm for Euclidean Matroid Supplier (for any $\epsilon > 0$). Define $c := \frac{2 \pi}{\cos^{-1}(1-\frac{\epsilon}{2})}$.  Given  any instance ${\cal H}$ of 1-in-3-SAT, we  generate an instance ${\cal E}$ of Euclidean Matroid Supplier as follows. Let the variables in ${\cal H}$ be $x_1, \ldots, x_n$, and suppose it has $m$ clauses. Let $d$ be an integer with $d \geq \max(\frac{c+1}{4}, m)$. In ${\cal E}$, we create $n$ cycles embedded as regular $4d$-gons of unit side length, with each cycle $C_i$ representing   variable $x_i$. The  cycles are placed far apart so that no vertex is within distance 3 of a vertex from a different cycle. For each cycle $C_i$, we label its vertices alternatively as  clients and  suppliers. Moreover, the suppliers on cycle $C_i$ are   alternatively labeled  as $x_i$ or $\neg x_i$. More precisely, if the vertices on $C_i$ are numbered   $j=1, 2, \ldots, 4d$ then we label the vertices as follows:
\begin{equation*}
    f(j) = \begin{cases}
              x_i  \text{ (supplier)} & \text{if } j \equiv 0 \mod 4,\\
              \neg x_i \text{ (supplier)} & \text{if } j \equiv 2 \mod 4,\\
            c_{ij} \text{ (client)} & \text{otherwise}.
          \end{cases}
\end{equation*}

Note that the number of suppliers in each cycle is $2d$, leading to $2nd$ suppliers in total.  Let $I$ denote the set of all suppliers. Now we  construct a partition matroid over $I$ in the following way. For each clause $k\in [m]$, say involving variables $x_{i_1}$, $x_{i_2}$, $x_{i_3}$, part $P_k\sse I$ consists of one supplier each from cycles $C_{i_1}, C_{i_2}, C_{i_3}$, where we take a supplier labeled $x_{i_j}$ (resp.  $\neg x_{i_j}$) if the clause  uses $x_{i_j}$ (resp. $\neg x_{i_j}$). See Figure~\ref{fig:matroid} for an 
example. 
We ensure that each supplier is in at most one part:  that is possible because each cycle contains $d\ge m$ suppliers of each label. Finally,  we gather all suppliers not  in any part $P_k$  into another part $P_0=I\setminus \left(\cup_{k=1}^m P_k\right)$. The partition matroid is required to pick at most one supplier from each part $\{P_k\}_{k=1}^m$ and at most  $dn - m$ suppliers from part $P_0$. 
\begin{figure}
    \centering
    \includegraphics[scale=0.14]{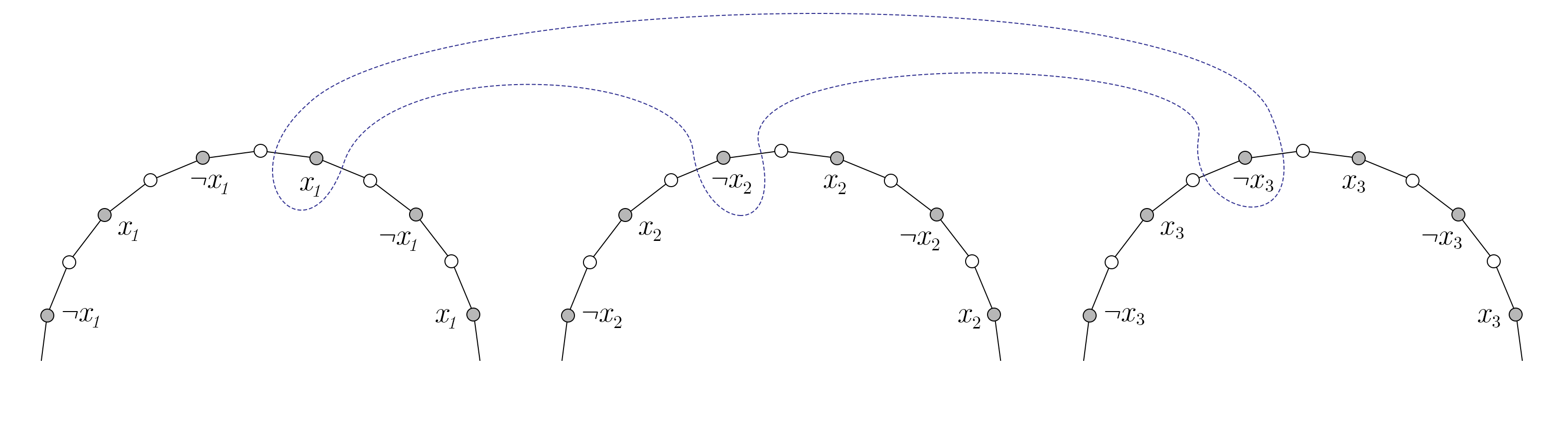}
    \caption{Example of the part corresponding to clause $x_1 \vee \neg x_2 \vee \neg x_3$. \label{fig:matroid}}
\end{figure}

\paragraph{Yes case.}
Suppose  that the  1-in-3-SAT instance ${\cal H}$ is satisfiable by some assignment $a=\{a_i\}_{i=1}^n$ of variables. Consider the Matroid Supplier solution $S$ that selects from each cycle $C_i$ all the $x_i$ (resp. $\neg x_i$) suppliers if $a_i=true$ (resp. $a_i=false$). The total number of selected suppliers $|S|=dn$. Note that each client is within distance one from some supplier in $S$. Moreover, for each clause $k\in [m]$, exactly one literal of this clause is true in assignment $a$: this implies that $|S\cap P_k|=1$.  It follows that a total of $m$ suppliers are selected from $\cup_{k=1}^m P_k$, which means $|S\cap P_0|=nd-m$. Hence, $S$ satisfies the partition matroid constraint. So, the optimal value of instance ${\cal E}$ is at most $1$.

\paragraph{No case.} Suppose that $S'$ is a solution to Matroid Supplier of objective at most $3- \epsilon$. Note that the distance between any client and supplier is either $1$ or at least $1 + 2 \cos(\pi - \frac{4d-2}{4d}\pi) > 1 + 2 \cos(\frac{2}{c}\pi) = 3 - \epsilon$.   So the objective value of solution $S'$ must be one.  
 \begin{cl}
Consider any solution $S'$ to ${\cal E}$ with objective  $1$. For each $i \in [n]$, $S'$ contains either all the $x_i$ suppliers or all the $\neg x_i$ suppliers in cycle $C_i$. Moreover, $|S'\cap P_k| = 1$ for all $k\in [m]$.
\end{cl}
\begin{proof} By the matroid constraint it is clear that $|S'|\le dn$. Note that  each supplier is at unit distance from at most 2 clients, and each cycle has $2d$ clients. Therefore, solution $S'$ must contain at least $d$ suppliers in each cycle $C_i$. As there are $n$ cycles, we must have  $|S'|= dn$, and the first statement  follows. To see the second statement, note that the only way we can have $|S'|= dn$ is to pick exactly one supplier from each $\{P_k\}_{k=1}^m$. \end{proof}
Now, consider the assignment $a'_i=true$ if $S'$ contains all the $x_i$-suppliers in cycle $C_i$, and $a'_i=false$ otherwise. For each clause $k\in [m]$, we  have exactly one true literal in $a'$ because $|S'\cap P_k|=1$. So $a'$ is a valid assignment for instance ${\cal H}$.

Therefore, if ${\cal H}$ is unsatisfiable, the optimal value of ${\cal E}$ is more than $3 - \epsilon$. 
Theorem~\ref{thm:hardness} now follows.

{
\bibliographystyle{alpha}
\bibliography{references}
}

\end{document}